\documentclass[11pt]{article}

\usepackage{fancyhdr}
\usepackage{balance}
\usepackage{graphicx}
\usepackage{acronym}
\usepackage{amsmath} 
\usepackage{amsthm}
\newtheorem{theorem}{Theorem}
\usepackage{subfigure}
\newcommand{\comments}[1]{} 
\usepackage{float}

\newcommand{\coo}{\rm{CO_2}} 

\acrodef{OPEX}[OPEX]{Operating Expenses}
\acrodef{UE}[UE]{User Equipment}
\acrodef{BS}[BS]{Base Station}
\acrodef{DTX}[DTX]{Discontinuous Transmission}
\acrodef{PAPR}[PAPR]{Peak-to-Average Power Ratio }
\acrodef{SC-FDMA}[SC-FDMA]{Single-carrier FDMA}
\acrodef{FDMA}[FDMA]{Frequency Division Multiple Access}
\acrodef{TDMA}[TDMA]{Time Division Multiple Access}
\acrodef{CDMA}[CDMA]{Code Division Multiple Access}
\acrodef{OFDMA}[OFDMA]{Orthogonal Frequency Division Multiple Access}
\acrodef{ICT}[ICT]{Information and Communication Technologies}
\acrodef{QoS}[QoS]{Quality of Service}
\acrodef{PA}[PA]{Power Amplifier}
\acrodef{RS}[RS]{Resource Sharing}
\acrodef{PC}[PC]{Power Control}
\acrodef{SOTA}[SotA]{State-Of-The-Art}
\acrodef{EE}[EE]{Energy Efficiency}
\acrodef{SINR}[SINR]{Signal-to-Interference-and-Noise-Ratio}
\acrodef{LTE}[LTE]{Long Term Evolution}
\acrodef{EARTH}[EARTH]{Energy Aware Radio and neTwork tecHnologies}
\acrodef{MIMO}[MIMO]{Multiple-Input and Multiple-Output (transmission)}
\acrodef{SISO}[SISO]{Single-Input and Single-Output (transmission)}
\acrodef{RE}[RE]{Resource Element}
\acrodef{SNR}[SNR]{Signal-to-Noise-Ratio}
\acrodef{ACLR}[ACLR]{Adjacent Carrier Leakage Ratio}
\acrodef{PRAIS}[PRAIS]{Power and Resource Allocation Including Sleep}
\acrodef{CSI}[CSI]{Channel State Information}
\pdfoutput=1

\begin{document}

%
\title{Downlink Power Control Minimizing Base Station Energy Consumption}

\author{{Hauke~Holtkamp, Gunther~Auer}\vspace{3mm}\\
DOCOMO Euro-Labs\\
D-80687 Munich, Germany \\Email:
\{holtkamp, auer\}@docomolab-euro.com
\and Harald~Haas\vspace{2mm}\\
Institute for Digital Communications\\
Joint Research Institute for Signal and Image Processing\\
 The University of Edinburgh,
EH9 3JL, Edinburgh, UK\\ E-mail: h.haas@ed.ac.uk}


\acresetall


\maketitle

\begin{abstract}
We consider single cell multi-user OFDMA downlink resource allocation such that average supply power is minimized while fulfilling a set of target rates. Available degrees of freedom are transmission power and duration. This paper extends our previous work on power optimal resource allocation in the mobile downlink by detailing the optimal power control strategy investigation and extracting fundamental characteristics of power optimal operation in cellular downlink. The allocation strategy that minimizes overall power consumption requires the transmission power on all links to be increased if only one link degrades. Furthermore, we show that for mobile stations with equal channels but different rate requirements, it is power optimal to assign equal transmit powers with proportional transmit durations. To relate the effectiveness of power control to live operation, we consider different variants of the power model which maps transmit power to supply power. We show that due to the affine mapping, the solution is independent of the power model. However, the effectiveness of power control measures is completely dependent on the underlying hardware and the load dependence factor of a base station (instead of absolute consumption values). Finally, we conclude that power control measures in base stations are most relevant in macro stations which have a load dependence factor exceeding 50\%.

\end{abstract}

\section{Introduction}
Mobile traffic volume is growing at a rapid pace. This requires more \acp{BS} which are more powerful and more densely deployed. The operation of todays' mobile networks causes $0.3\%$ of global $\coo$ emissions and $80\%$ of the operating energy is spent at the radio \ac{BS} sites causing significant electricity and diesel bills for network operators~\cite{fmbf1001, aghimfbhz1001}. These combined environmental concerns and rising energy costs provide strong incentives to reduce the power consumption of future \ac{BS} generations.

Traffic statistics reveal that network load varies significantly over the course of a day and different deployment patterns such as rural, suburban or urban. While \acp{BS} operate at maximum efficiency during peak hours, their load adaptability is limited resulting in low energy efficiency in low load situations. We aim at actively decreasing the power consumption via power control during low loads when the spectral and computational resources far outweigh the traffic demand.

Recent studies show that the overall power consumption of a radio transceiver is dominated by the consumption of the \ac{PA} \cite{fbzfgjt1001}. Reduction of transmit power causes a significant decrease of the consumed power in the \ac{PA}. This general notion is analytically included in this paper via the EARTH\footnote{EU funded research project EARTH (Energy Aware Radio and neTwork tecHnologies), FP7-ICT-2009-4-247733-EARTH, Jan.~2010 to June~2012. https://www.ict-earth.eu} power model, which determines the overall \ac{BS} power consumption on the basis of transmit power~\cite{czbfjgav1001}. In this fashion, it is possible to determine the achieved absolute savings of power control strategies.

In our previous work, we have considered power control as an added improvement to systems which are capable of sleep modes~\cite{hah1101}. We evaluated the combined gains on several generations of base stations. In this work, we provide a detailed investigation of the optimal stand-alone power control strategy when \acp{BS} are not capable of sleep modes.

Power control has been extensively studied as a tool for rate maximization when spectral resources are sparse~\cite{g9701, ssha0801a}. But these strategies need to be adjusted when maximizing the spectral efficiency is not the objective function. Wong et al.~\cite{wclm9901} provide an algorithmic strategy for power allocation in multiuser \ac{OFDMA}, but only for static modulation. By employing the Shannon limit directly, we presume optimal modulation. Al-Shatri et al.~\cite{aw1001} optimize the operating efficiency, but require a predetermined static number of resources per user. Cui et al.~\cite{cgb0501} optimize the modulation scheme on the link level, which we extend to the \ac{BS} level by incorporating the EARTH power model. This allows for minimizing the base station supply power, rather than the output power radiated at the antenna elements.. 


The remainder of this paper is structured as follows. The system model is presented in Section~\ref{system} which results in the optimal allocation strategy. In Section~\ref{methodology}, first the optimal allocation of transmission power and time are discussed. It is then derived what transmit powers users with equal channels should receive. The power model is added in Section~\ref{powermodel} to change the optimization variable from transmit power to supply power. Conclusions are drawn in Section~\ref{conclusion}.

\section{System model}
\label{system}
When a number of bits transmitted on a link with fixed bandwidth and optimal modulation is to be increased, there are two options. Either increase the link rate by raising the transmission power or transmit for a longer time. From a power efficiency perspective providing a fixed rate is thus a trade-off between transmit power and transmission duration. As illustrated in Figure~\ref{PC}, a higher power for smaller duration will provide the same rate as a lower power with higher transmission duration. For a single link, transmission should always be the longest to allow for the lowest transmit power. However, in a shared multi-user channel with orthogonal access, all links have to be considered which each have individual rate requirements that have to be fulfilled in a set time.

\begin{figure}
\begin{center}
\includegraphics[width=0.6\textwidth]{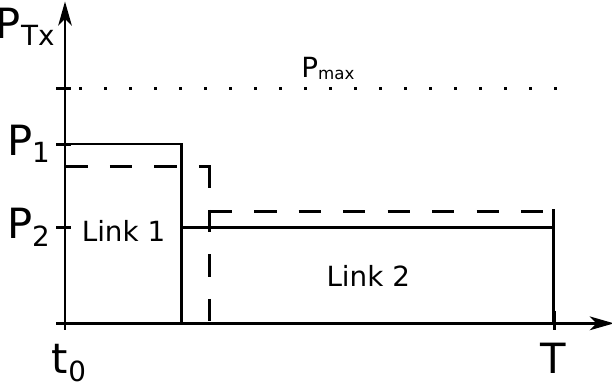}
\end{center}
\caption{Illustration of two possible power/time trade-offs that provide equal rates on both links.}
  \label{PC}
\end{figure}

We proceed with a derivation of the optimal allocation strategy. We employ the Shannon bound to map transmission power to achievable rate as the most fundamental law of energy consumption in communications. We assume that in a system a set of mobiles with known \ac{CSI} has a set of rate requirements that needs to be fulfilled at minimum average system-wide power consumption. 

A cell consists of a \ac{BS} and $N_{\rm{L}}$ links (mobiles). The rate per link $i$ is upper bounded by
\begin{equation}
R_i = W \log_2 \left( 1 + \gamma_i \right),
\label{eq:Shannon}
\end{equation}
where $W$ is the channel bandwidth in Hz and $\gamma_i = \frac{G_{i} P_i}{N_0}$ is the \ac{SNR} with $G_{i}$ the link channel gain, $P_i$ the transmit power on link $i$ in W and $N_0$ thermal noise in W. The noise power is defined by $N_0 = W k \vartheta$ with Boltzmann constant~$k$ and operating temperature~$\vartheta$ in Kelvin. While rate and bandwidth are linearly related, rate and transmit power have a logarithmic relationship. As a consequence it is much more expensive in terms of power to increase channel rate than in terms of bandwidth. In other words, if there is a choice between leaving idle bands and transmitting at higher power and using all available bands and transmitting at the lowest required power, then the latter will always consume less overall transmit power. 

The average target rate $\overline{R}_i$ per link in bps has to be fulfilled within a total time frame $T$. Total energy consumed by the system is defined as 
\begin{equation}
\label{esys}
E_{\rm{SYS}} = \overline{P}_{\rm{SYS}} T
\end{equation}
in J, where $\overline{P}_{\rm{SYS}}$ is the system average power consumption. 

Normalized transmission time per link is given by
\begin{equation}
 \label{eq:mu}
  \mu_i = \frac{t_i}{T}
\end{equation}
where $t_i$ is the transmission time on link $i$ in seconds and $\mu_i > 0$. Optimization of normalized time results in average power minimization and is more illustrative than energy minimization (which is only meaningful for a known $T$). 

The average target rate on link~$i$ over the time slot~$T$ depends on the transmission time $\mu_i$ and the rate during transmission, $R_i$:
\begin{equation}
 \overline{R}_i = \mu_i R_i.
\end{equation}

We are now able to find the transmission power per link as a function of the required average rate: 
\begin{equation}
\label{eq:ptxavg}
 \overline{P}_{\rm{Tx},\it{i}}(\overline{R}_i) = \frac{N_0}{G_i} \left( 2^\frac{\overline{R}_i}{W \mu_i} - 1 \right),
\end{equation}
where $0 < \overline{P}_{\rm{Tx},\it{i}}(\overline{R}_i) < P_{\rm{max}}$ for some $P_{\rm{max}}$.

To account for the fact that all links are served by the \ac{BS} orthogonally on the shared resource, the system average transmission power at the base station for all links over $T$ is the sum of individual transmit powers weighted with the transmit duration~\cite{hah1101}
\begin{equation}
\label{PSYS}
\begin{aligned}
 \overline{P}_{\rm{SYS}}(\overline{R}_i)  &= \sum_{i=1}^{N_{\rm{L}}} \mu_i  \cdot P_{\rm{Tx},i}(\overline{R}_i)  \\
					  &= \sum_{i=1}^{N_{\rm{L}}} \mu_i  \frac{N_0}{G_i} \left( 2^\frac{\overline{R}_i}{W \mu_i} - 1  \right)
\end{aligned}
\end{equation}
where all links have to be served in the available time $T$. The combined duration of all transmissions must be less or equal to $T$. But since it is clearly most efficient to use the entire available $T$, it holds that
\begin{equation}
\label{muconstraint}
\displaystyle\sum_{i=1}^{N_{\rm{L}}} \mu_i = 1.
\end{equation}

The allocation vector of transmission durations which minimizes \eqref{PSYS} is power optimal. 

\section{Methodology and Results}
\label{methodology}	\fancyhead[l]{
		\includegraphics[width=23mm,height=23mm]{FutureNetworkAndMobileSummitLogo}
		\tiny \hspace*{1mm} \\
	}

\subsection{On transmission power and transmission time}

For a particular target rate on a link, a lower transmit time results in a higher transmission power and vice versa. Since transmission power is capped by $P_{\rm{max}}$, this in turn lower bounds $\mu_i$. If the target rates cannot be fulfilled without violation of $P_{\rm{max}}$ on all links (or if $\sum_{i} \mu_i > 1$), then the system is overloaded. Thus, the optimal allocation is a trade-off of transmission times between links which depend on the individual channel gains and target rates. First, we inspect the behavior of transmit powers and times as they depend on the channel gains.

For illustration, the individual optimal transmission powers and times of a system with two links and equal target rates are plotted in Figure~\ref{fig:DiffoverG1G21bps} with the parameter set in Table~\ref{tab:uniformdistrscenarioparameters}. In this setting, the channel gain on link $2$ is constant, while it is sweeped over a range of $100$~dB for link~$1$. Plotted in Figure~\ref{fig:DiffP1P2overG1G21bps} are the system optimal transmit powers. The x-axis contains the ratio of the two channel gains (which equals the difference in dB) to emphasize the relative gap in channel gains. When channel gains are equal the power optimal allocation is clearly to assign equal powers. As link~$2$ degrades, more power should be allocated to it. The important finding is here that for optimal power allocation, transmit powers have to be increased on both links as channel gain difference increases. It is not possible to minimize \ac{BS} power consumption by considering individual links. Rather, all links have to be considered simultaneously.

\begin{table}
      \caption{Simulation Parameters}
\centering
      \begin{tabular}{c|c}
	Parameter          			& Value\\
	\hline
	Bandwidth $W$ 				& 10\,MHz \\
	Thermal noise 				& $-103\,$dBm \\
	Static channel gain 			& $-100\,$dB \\
	\end{tabular}
\label{tab:uniformdistrscenarioparameters}
\end{table}

Figure~\ref{fig:Diffmu1mu2overG1G21bps} shows the corresponding transmission times. It can be seen that transmission times should compensate for the change in transmit powers on each link. The target rates are fulfilled by appropriate selection of transmission times $\mu_i$. For example, in a cellular system, this occurs in the downlink case between a cell edge user with low channel gain on link~$1$ and a center user with high channel gain on link~$2$. The difference in transmission power in Figure~\ref{fig:DiffP1P2overG1G21bps} shows that at higher differents in \ac{SNR} there is less flexibility in the selection of $\mu_i$ resulting in a smaller difference of transmission powers.

\begin{figure}[ht]
\centering
\subfigure[Optimal transmission powers in a 2-link system as a function of channel gain difference. The difference in powers $\Delta P$ in dB is equivalent to the ratio in linear terms.]{
\includegraphics[width=0.6\textwidth]{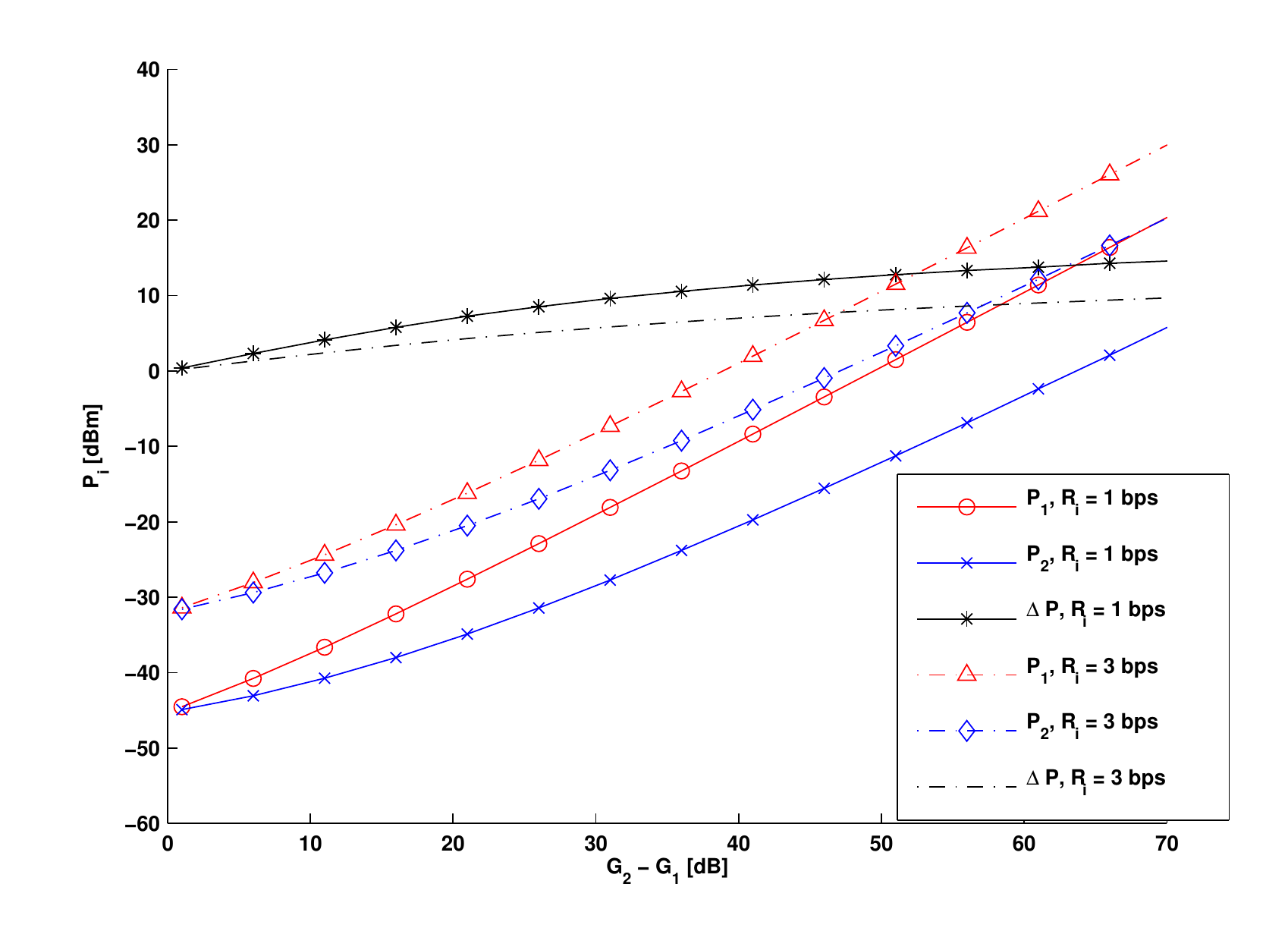}
\label{fig:DiffP1P2overG1G21bps}
}
\subfigure[Optimal transmission times in a 2-link system as a function of channel gain difference.]{
\includegraphics[width=0.6\textwidth]{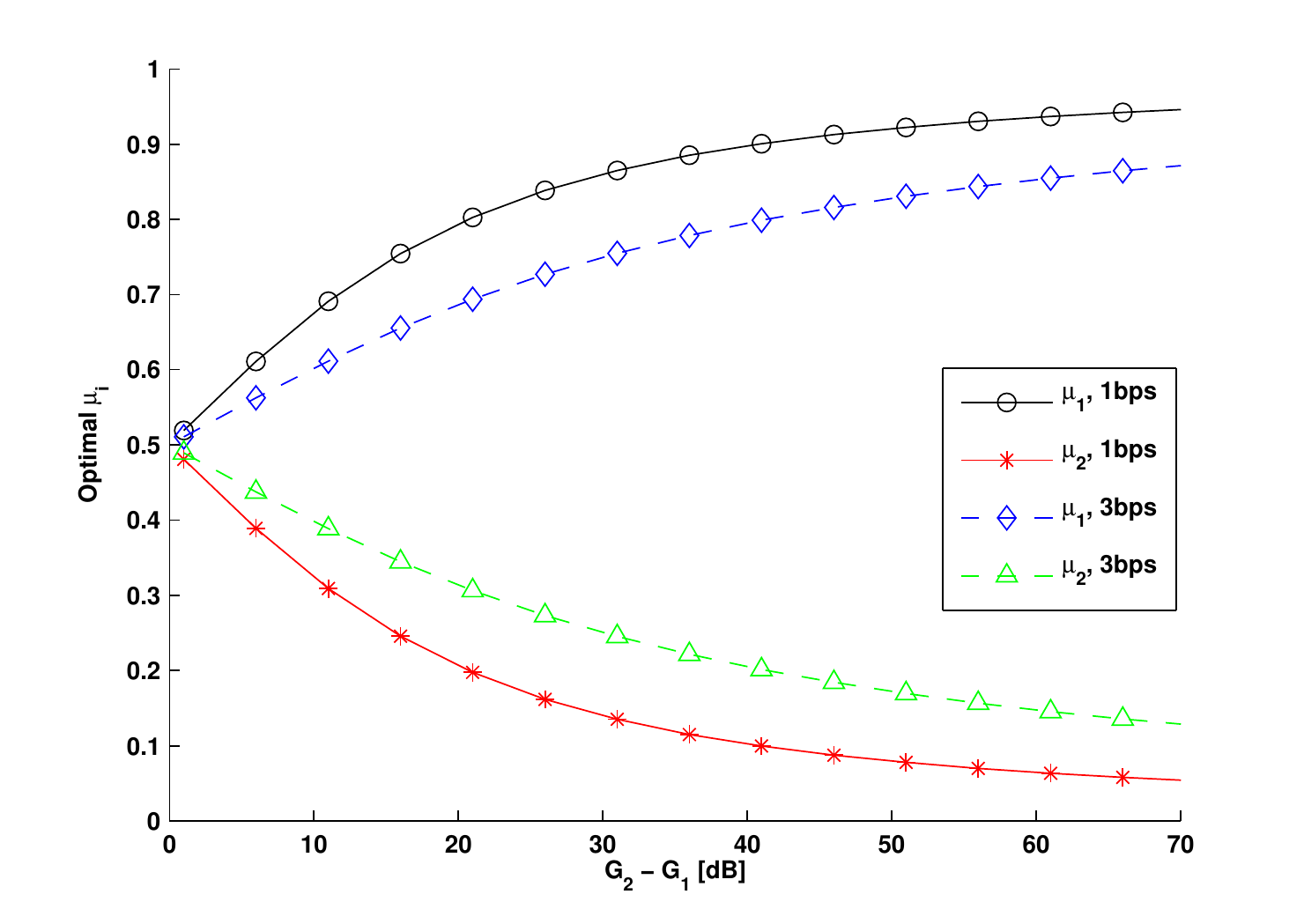}
\label{fig:Diffmu1mu2overG1G21bps}
}
\caption{Illustration of optimal transmission powers and times.}
\label{fig:DiffoverG1G21bps}
\end{figure}

For more than two users, the same holds true. If one link in the system degrades, all links should increase their power. To uphold the rate requirements, the unchanged links will reduce transmission time and the bad link receives more transmission time. 

%
%

\subsection{Optimal strategy on equal channels}

When considering the system which depends on the given channel gains and target rates, finding the optimal transmission times $\mu_i$ is an analytically difficult problem for which to our knowledge no closed form solution exists~\cite{hah1101}. However, there exists a special case, which can be analytically solved. Suppose all links have equal channel gains $G_i = G$, then \eqref{PSYS} simplifies to
\begin{equation}
\label{PSYS2}
 \overline{P}_{\rm{SYS}}(\overline{R}_i)  = \frac{N_0}{G} \sum_{i=1}^{N_{\rm{L}}} \mu_i \left( 2^\frac{\overline{R}_i}{W \mu_i} - 1 \right).
\end{equation}

This new problem can be solved using a Lagrange multiplier.

\begin{theorem}
 \label{theorem}
\begin{equation}
 \mu_i = \frac{R_i}{\sum_{i=1}^{N_{\rm{L}}} R_i}
\end{equation}
is the minimizer of 
\begin{equation}
\label{PSYS3}
 \overline{P}_{\rm{SYS}}(\overline{R}_i)  = \sum_{i=1}^{N_{\rm{L}}} \mu_i \left( 2^\frac{\overline{R}_i}{W \mu_i} - 1 \right).
\end{equation}

\end{theorem}

\begin{proof}
The partial derivatives of \eqref{PSYS3} and \eqref{muconstraint} are
$$
\frac{\partial\overline{P}_{\rm{SYS}}(\overline{R}_i)}{\partial \mu_i} = 2^{\frac{R_i}{W \mu_i}}-1-(\frac{R_i}{W \mu_i})\log(2)2^{\frac{R_i}{W \mu_i}}
$$
and unity, respectively; $\log$ refers to the natural logarithm.

Therefore, the gradients of the Lagrange function $\Lambda(\mu_i, \lambda)$ for any $i$ are
\begin{equation}
  \nabla_{\mu_i, \lambda} \Lambda(\mu_i, \lambda)= \frac{\partial\overline{P}_{\rm{SYS}}(\overline{R}_i)}{\partial \mu_i} + \lambda.
\end{equation}

The partial derivatives are independent of each other and only depend on $\lambda$. Thus, they have to be equal.

Functions of the type $x\mapsto 2^{x}(1-x\log(2))$ are monotone on $x\ge0$, hence there cannot be two different arguments for one $i$ yielding the same function value. That is, one must have $\mu_i=cR_i$ for a given $c$. Since $\sum\mu_i=1$, $1/c$ is the sum of $R_i$ over $i$ and the optimum is given by
\begin{equation}
\label{musolution}
\mu_i=\frac{R_i}{R_1+\cdots+R_{N_{\rm{L}}}}.
\end{equation}
\end{proof}

This means that transmission times are directly proportional to the link target rates. The optimal transmission power on each link can be directly found from the sum of target rates using \eqref{eq:ptxavg} and \eqref{musolution}. Optimization of mobile user subsets can be achieved by assigning mobiles in groups according to their channel quality (or distance from the \ac{BS}) or by appropriate scheduling in time of users with similar link quality.

%

\section{Power control strategy in the power model}
\label{powermodel}
The power model~\cite{aggsoigdb1101} is a detailed representative model of how such \ac{BS} components as radio transceiver, baseband interface, power amplifier, AC-DC-converter, DC-DC-converter and cooling fans act together forming an overall power consumption behavior. It was found that although highly complex in detail, the combined consumption of today's \acp{BS} can be represented by an affine function. There is a static comsumption that applies independent of load, $P_0$, which we refer to as the idle power consumption. In addition, there is consumption which depends on the power delivered to the antenna. The rate of increase is represented by a load factor~$l$. The model is analytically represented by 
\begin{equation}
 \label{psupplymax}
P_{\rm{supply}} = P_0 + l P_{\rm{Tx}},
\end{equation}
with $P_{\rm{Tx}} \leq P_{\rm{max}}$ and is illustrated in Figure~\ref{fig:PM}.

As a metric for the share of load-dependent consumption of the overall consumption, we define the load dependence of a \ac{BS} as
\begin{equation}
 \eta_{\rm{ld}} = \frac{l P_{\rm{max}}}{P_0+l P_{\rm{max}}}.
\end{equation}

\begin{figure}
\centering
 \includegraphics[width=0.6\textwidth]{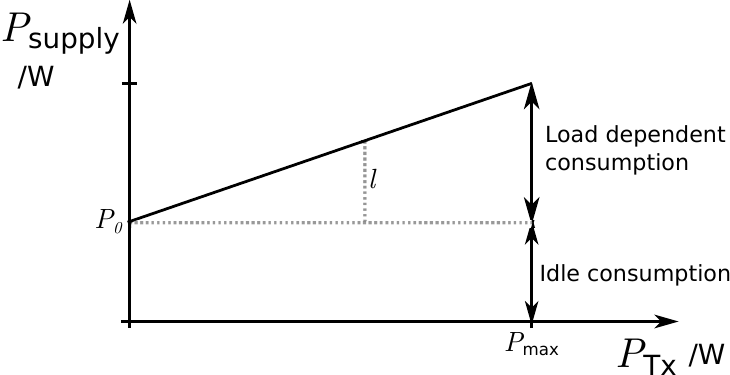}
\caption{Illustration of the power model with load dependent and idle consumption.}
\label{fig:PM}
\end{figure}

It is important to consider power control and transmission power in relationship to the consumption of the entire \ac{BS}. The power model provides this by mapping transmission power to supply power. Instead of the transmission power, we are now able to minimize the supply power consumption. Addition of the power model to the power control strategy changes \eqref{PSYS} in the following fashion:

\begin{equation}
\label{PSYSPM}
\begin{aligned}
 \overline{P}_{\rm{supply}}(\overline{R}_i)  &= \sum_{i=1}^{N_{\rm{L}}} \mu_i \left( P_0 + l P_{\rm{Tx},i}(\overline{R}_i)  \right)\\
					     &= \sum_{i=1}^{N_{\rm{L}}} \mu_i \left( P_0 + l \frac{N_0}{G_i} \left( 2^\frac{\overline{R}_i}{W \mu_i} - 1 \right) \right)
\end{aligned}
\end{equation}

Although not intuitive, it turns out that addition of the power model does not affect the power control strategy. Extending the Lagrange multiplier analysis to the power model inclusive problem, the partial derivative of the cost function now becomes
\begin{equation}
 \frac{\partial\overline{P}_{\rm{SYS}}(\overline{R}_i)}{\partial \mu_i} = P_0 + l \frac{N_0}{G_i} \left( 2^{\frac{R_i}{W \mu_i}}(1-(\frac{R_i}{W \mu_i})\log(2))-1 \right).
\end{equation}

In line with the proof for Theorem~\ref{theorem}, it still holds that all partial derivatives are equal. Hence, $P_0$ and $l$ have no influence on the solution. This bears the important consequence that the power control strategy is independent of the underlying power model, and thus the underlying \ac{BS}. It is therefore optimal for all \acp{BS}.

However, although the strategy is independent of the hardware, the benefits of power control strongly depend on hardware. See Figure~\ref{fig:PMcomparisonOverRate} for an illustration in a cell with ten users. Only when the load dependence factor~$\eta_{\rm{ld}}$ of a \ac{BS} is high enough, can power control provide relevant savings. The effectiveness of power control in future \acp{BS} therefore strongly depends on hardware developments. For comparison, we consider typical load dependence factors of different \ac{BS} types in Table~\ref{tab:ldfactors}. Due to the low range and low transmit powers of smaller \ac{BS} types, they have load dependence factors much smaller than $50\%$. Therefore, power control has limited applicability in smaller \ac{BS} types than in macro stations.   

\begin{figure}
\centering
 \includegraphics[width=0.6\textwidth]{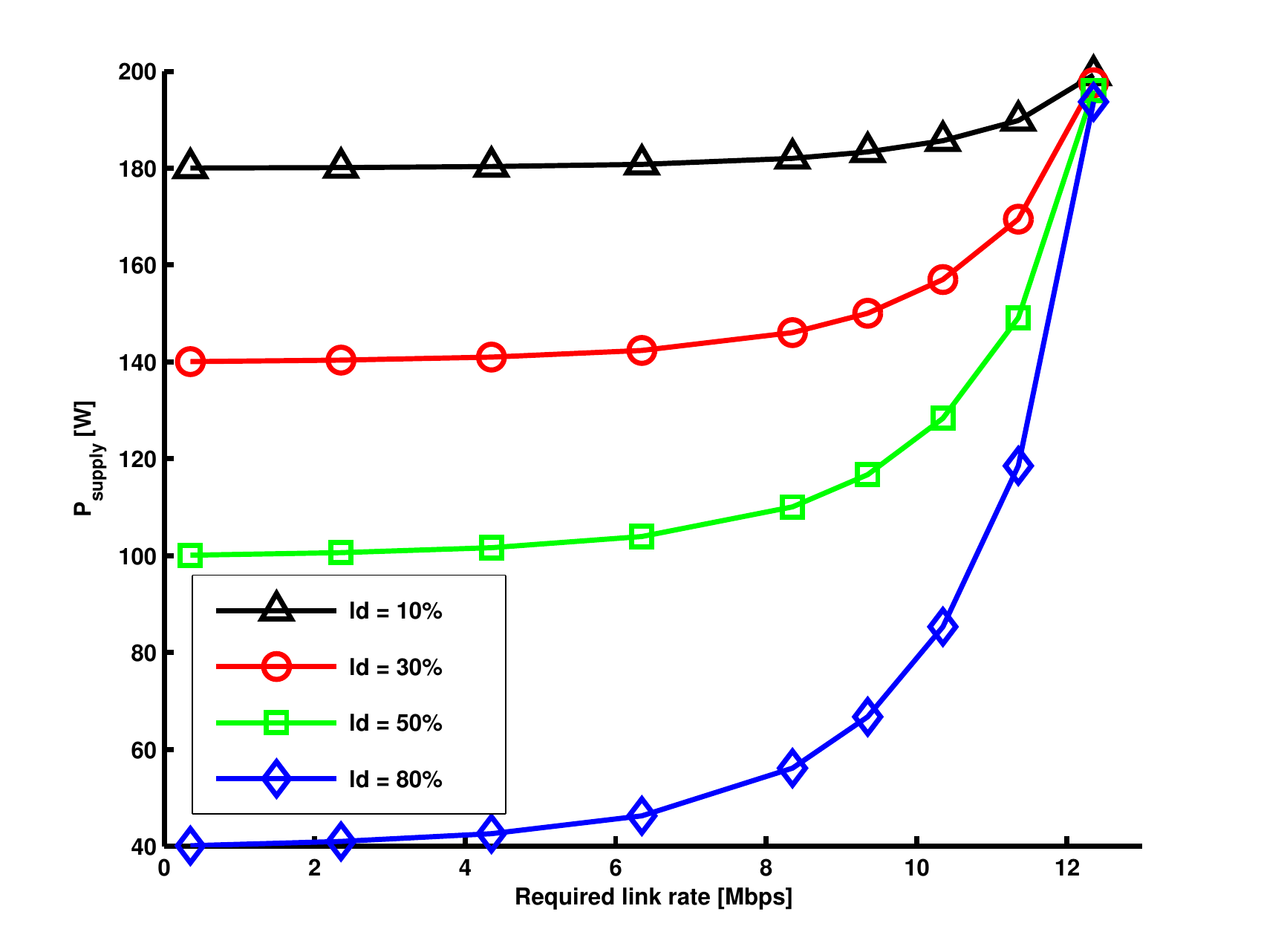}
\caption{Comparison of power models on the effectiveness of power control.}
\label{fig:PMcomparisonOverRate}
\end{figure}

\begin{table}
      \caption{Typical load dependence factors by \ac{BS} type~\cite{czbfjgav1001}}
\centering
      \begin{tabular}{c|cc}
	\ac{BS} type          		&$P_{\rm{max}}$	& Load dependence factor~$\eta_{\rm{ld}}$\\
	\hline
	Macro	 			&$46$~dBm		& $50\%$\\
	Micro	 			&$38$~dBm		& $30\%$\\
	Pico	 			&$21$~dBm		& $14\%$\\
	Femto	 			&$17$~dBm		& $10\%$\\
	\end{tabular}
\label{tab:ldfactors}
\end{table}

\section{Conclusion}
\label{conclusion}
In this paper, supply power optimal transmission power control is studied in the multi-user downlink setting where the \ac{QoS} criterion is a target link rate. The minimization variable is the average supply power of a \ac{BS} as a function of the required rates and channel gains. The system can be analytically represented. The first important finding is that in power optimal allocation, the system has to be considered as a whole. If a channel in the system degrades, the transmission powers for all links have to be increased while transmission durations for all links except the degrading link are decreased. While this convex problem does not have a closed-form solution, it can be simplified if only users with equal channel gains are considered at one time. It has been shown that in that case the optimal transmission time is proportional to the required link rate over the sum of link rates and the optimal transmission powers are equal for all links. To account for realistic hardware consumption, a power model is considered which represents static and load dependent consumption of a \ac{BS}. The addition of the power model allows estimation of absolute consumption values rather than transmission power. Inclusion of the power model into the power allocation problem reveals that the optimal strategy holds independent of the chosen power models. However, only if the power consumption of a \ac{BS} is largely load-dependent, does power control result in significant savings. Otherwise, the idle consumption is several magnitudes larger than the effect of transmission power. Since \ac{BS} types other than macro stations have much lower transmission powers and thus lower load dependence factors, power control is most effective in macro stations.

\balance

\section*{Acknowledgements}

This work has received funding from the European Community's $7^{\textrm{th}}$ Framework Programme [FP7/2007-2013. EARTH, Energy Aware Radio and neTwork tecHnologies] under grant agreement n${}^\circ$ 247733.

The authors gratefully acknowledge the invaluable insights and visions received from partners of the EARTH consortium.


\bibliography{../../../../../../DOCOMO/reference/DOCOMO,../../../../../../DOCOMO/reference/general}

\begin{thebibliography}{10}

\bibitem{fmbf1001}
A.~Fehske, J.~Malmodin, G.~Bicz\'ok, and G.~Fettweis, ``{The Global Carbon
  Footprint of Mobile Communications - The Ecological and Economic
  Perspective},'' {\em {IEEE Communications Magazine}}, 2010.

\bibitem{aghimfbhz1001}
G.~Auer, I.~G\'odor, L.~H\'evizi, M.~A. Imran, J.~Malmodin, P.~Fasekas,
  G.~Bicz\'ok, H.~Holtkamp, D.~Zeller, O.~Blume, and R.~Tafazolli, ``{Enablers
  for Energy Efficient Wireless Networks},'' in {\em {Proc. of the Vehicular
  Technology Conference (VTC)}}, 2010.

\bibitem{fbzfgjt1001}
D.~Ferling, T.~Bohn, D.~Zeller, P.~Frenger, I.~G\'odor, Y.~Jading, and
  W.~Tomaselli, ``{Energy Efficiency Approaches for Radio Nodes},'' in {\em
  Future Network \& Mobile Summit 2010}, 2010.

\bibitem{czbfjgav1001}
L.~Correia, D.~Zeller, O.~Blume, D.~Ferling, Y.~Jading, I.~G\'odor, G.~Auer,
  and L.~Van Der~Perre, ``{Challenges and Enabling Technologies for Energy
  Aware Mobile Radio Networks},'' {\em Communications Magazine, IEEE}, vol.~48,
  no.~11, pp.~66 --72, 2010.

\bibitem{hah1101}
H.~Holtkamp, G.~Auer, and H.~Haas, ``{On Minimizing Base Station Power
  Consumption},'' in {\em {Proceedings of the IEEE VTC 2011-Fall}}, 2011.

\bibitem{g9701}
A.~J. Goldsmith, ``{The Capacity of Downlink Fading Channels with Variable Rate
  and Power},'' {\em {IEE Transactions on Vehicular Technology}}, vol.~46,
  pp.~569--580, 1997.

\bibitem{ssha0801a}
S.~Sinanovi\'{c}, N.~Serafimovski, H.~Haas, and G.~Auer, ``{Maximising the
  System Spectral Efficiency in a Decentralised 2-link Wireless Network},''
  {\em Eurasip Journal on Wireless Communications and Networking}, vol.~2008,
  p.~13, 2008.
\newblock doi:10.1155/2008/867959.

\bibitem{wclm9901}
C.~Y. Wong, R.~S. Cheng, K.~B. Lataief, and R.~D. Murch, ``{Multiuser OFDM with
  Adaptive Subcarrier, Bit, and Power Allocation},'' {\em {IEEE} {J}ournal on
  {S}elected {A}reas in {C}ommunications}, vol.~17, pp.~1747--1758, Oct. 1999.

\bibitem{aw1001}
H.~Al-Shatri and T.~Weber, ``{Fair Power Allocation for Sum-Rate Maximization
  in Multiuser OFDMA},'' in {\em {Proceedings of the International ITG Workshop
  on Smart Antennas 2010}}, 2010.

\bibitem{cgb0501}
S.~Cui, A.~J. Goldsmith, and A.~Bahai, ``{Energy-Constrained Modulation
  Optimization},'' {\em {IEEE Transactions on Wireless Communications}},
  vol.~4, pp.~2349--2360, 2005.

\bibitem{aggsoigdb1101}
G.~Auer, V.~Giannini, I.~G\'odor, P.~Skillermark, M.~Olsson, M.~A. Imran, M.~J.
  Gonzalez, C.~Desset, O.~Blume, and A.~Fehske, ``{How Much Energy is Needed to
  Run a Wireless Network?},'' {\em IEEE Wireless Communications}, vol.~18,
  pp.~40 --49, Oct 2011.

\end{thebibliography}

\bibliographystyle{ieeetr}

\end{document}